\documentclass[conference, letterpaper, twocolumn, 10pt]{IEEEtran}

\usepackage[utf8]{inputenc}
\usepackage{verbatim}
\usepackage[cmex10]{amsmath}
\usepackage{amssymb}
\usepackage{mathrsfs}
\usepackage{graphicx}
\usepackage[numbers]{natbib}

\newtheorem{proposition}{Proposition}

\newtheorem{corollary}{Corollary}

\newcommand{\expectation}{\ensuremath{\mathbb{E}}}
\newcommand{\Expt}{\expectation}
\newcommand{\probability}{\ensuremath{\mathbb{P}}}
\newcommand{\Prob}{\probability}
\newcommand{\sign}{\text{sign}}
\usepackage{url}

\begin{document}

\title{Properties of the Polarization Transformations for the Likelihood Ratios of Symmetric B-DMCs}
\author{\IEEEauthorblockN{Mine Alsan}\\
\small\IEEEauthorblockA{Information Theory Laboratory\\
Ecole Polytechnique F\' ed\' erale de Lausanne\\
CH-1015 Lausanne, Switzerland\\
Email: mine.alsan@epfl.ch}
\normalsize}
\maketitle
\pagestyle{empty}
\thispagestyle{empty}
\IEEEpeerreviewmaketitle

\maketitle

\begin{abstract}
%\boldmath
In this paper we investigate, starting with a symmetric B-DMC, the evolution of various probabilities of the likelihood ratios of the synthetic channels created by the recursive application of the basic polarization transformations.
The analysis provides a new perspective into the theory of channel polarization initiated by Ar{\i}kan and helps us to address a problem related to approximating the computations of the likelihood ratios of the synthetic channels.
\end{abstract}

\begin{IEEEkeywords}
Channel polarization, polar codes, min-sum approximation
\end{IEEEkeywords}

\section{Introduction}
Polar coding is a recent technique introduced by Ar{\i}kan~\cite{1669570} as an appealing error correction method;  
this class of codes are proved to achieve the symmetric capacity of any binary discrete memoryless channel (B-DMC) using low complexity encoders and decoders, 
and their block error probability is shown to decrease exponentially in the square root of the block length~\cite{5205856}. 

The design of polar codes is based on a phenomenon called channel polarization. 
The notion makes reference to two extreme situations of communication over a noiseless (perfect) channel and completely noisy channel.
In~\cite{1669570}, Arıkan describes a recursive process under which independent copies of a given B-DMC $W: \mathcal{X} \rightarrow \mathcal{Y}$ can be combined to exhibit polarization.  
The basic building block of this recursion consists of two successive channel transformations 
$W^{-}: \mathcal{X} \rightarrow \mathcal{Y}^{2}$ and $W^{+}: \mathcal{X} \rightarrow \mathcal{Y}^{2}\times\mathcal{X}$, whose transition probabilities
are defined as
\begin{align*} 
  &W^{-}(y_1 y_2 \mid u_1) = \displaystyle\sum_{u_2\in\mathcal{X}} \frac{1}{2} W(y_1 \mid u_1 \oplus u_2) W(y_2 \mid u_2), \\
  &W^{+}(y_1 y_2 u_1\mid u_2) = \frac{1}{2} W(y_1 \mid u_1 \oplus u_2) W(y_2 \mid u_2). 
\end{align*}
Referred as the basic polarization transformations, these constitute the elements of the design leading to the low complexity structure of the codes.

To build the theory of polarization,~\cite{1669570} considers the properties of the above transformations related to the symmetric capacities of the channels. 
Defined as 
  \begin{equation} \label{eq:symcap}
  I(W) = \displaystyle\sum_{x, y} \frac{1}{2} W(y\mid x) \log{\frac{W(y\mid x)}{\frac{\displaystyle1}{\displaystyle2} W(y\mid 0) + \frac{\displaystyle1}{\displaystyle2} W(y\mid 1)}}, \nonumber  
 \end{equation}
 by now it is well known that these transformations~\cite{1669570} 
\begin{enumerate}
 \item[(i)] preserve the sum symmetric capacity:
 \begin{equation*}
 I(W^{-}) + I(W^{+}) = 2 I(W),
 \end{equation*}
 \item[(ii)] improve the channel in $W^{+}$ and worsen in $W^{-}$:
 \begin{equation*}
 I(W^{-}) \leq I(W) \leq \hspace{1mm} I(W^{+}). 
 \end{equation*}  
 \end{enumerate}
This last property confirms that the evolution is in the right direction towards polarization. The idea now is to apply the same basic channel transformations 
to the channels $W^{-}$ and $W^{+}$. As a result, four channels $W^{--}$, $W^{-+}$, $W^{+-}$, and $W^{++}$ are obtained. However,
one is no longer able to compare in general the parameters of these four channels in terms of rate, except the knowledge that
the channel $W^{++}$ is the best one and the channel $W^{--}$ is the worst one. Instead of worrying about ordering the channels after a few steps,
the theory is founded by analyzing the convergence properties of the polarization process obtained by applying the transformations to the synthesized $\pm$ channels in a long sequence of steps. 

Let $(\Omega, \mathcal{F}, P)$ be a probability space. Assume the random sequence $B_1, \dots, B_{n}$ is drawn i.i.d according to a Bernoulli distribution 
with probabilities equal to $\frac{\displaystyle1}{\displaystyle2}$. Let $\mathcal{F}_{n}$ be the $\sigma$-algebra generated by this Bernoulli sequence. Then the polarization process for a given channel $W$ is defined~\cite{5205856} as the random sequence of channels $\{W_{n}\}$ such that $W_{0} = W$ and
\begin{equation}
 W_{n+1} = \left\lbrace \begin{array}{lll}
                           W_{n}^{-}  & \hbox{if} \hspace{2mm} B_{n} = 0\\ W_{n}^{+} & \hbox{if} \hspace{2mm} B_{n} = 1
                          \end{array}
\right. \nonumber
\end{equation}
for $n \geq 0$. In the sequel, the random process $I_{n} = I(W_{n})$ is defined and~\cite{1669570} proves the process $\{I_{n}, \mathcal{F}_{n}\}$
\begin{enumerate}
 \item[(iii)] is a bounded martingale on the interval $[0, 1]$,
 \item[(iv)] converges a.s. to a random variable $I_{\infty}$ such that $\Expt{\left[ I_{\infty}\right] } = I_{0}$, 
 where $I_{\infty}$ takes values a.s. in $\{0, 1\}$. 
\end{enumerate}
These cited two properties prove the recursive application of the basic polarization transformations lead to channel polarization, see \cite[Theorem 1]{1669570}.

The goal of this paper is to analyze the convergence properties of various random processes associated with the channel polarization process, as the ones described for the symmetric capacity process,  
but related this time to the likelihood ratios of the synthesized $\pm$ channels. We first apply this knowledge to revisit the theory of channel polarization for symmetric B-DMCs.
Subsequently, we shift our attention to the performance of an approximation to the minus polarization transformation known as the min-sum approximation in the coding theory literature. We identify a structure sufficient to guarantee no performance loss is incurred by an approximation, and we argue
slight modifications to the `min-sum' approximation can improve the performance. 

The next section explores these results. The final section gives the conclusions.

\section{results}
Let $W: \mathcal{X} \to \mathcal{Y}$ be a symmetric B-DMC. We define the likelihood ratio of this channel as $L(y) = W(y|1)/W(y|0)$ for $y\in\mathcal{Y}$. 
Similarly for each $n \geq 0$, the likelihood ratios of the $2^n$ channels $W_{n}: \mathcal{X} \to \mathcal{Y}^{2^n} \times \mathcal{X}^{i-1}$, for $i= 1, \ldots, 2^n$ are denoted as $L_{n}(\mathbf{y})$
for $\mathbf{y}\in\mathcal{Y}^{2^n}$.
\subsection*{Properties of the polar transforms}
In \cite[Equations (74) and (75)]{1669570} Ar{\i}kan shows the synthetic channels' likelihood ratios follow a recursive structure alongside the polarization process.
For a symmetric B-DMC, one can assume the all zeros sequence is sent through the channel. In this case, the corresponding likelihood ratio process can be defined as
\begin{equation*}
 L_{n+1}(\mathbf{y_1y_2}) = \begin{cases}
				 L_{n}^{-}(\mathbf{y_1y_2}), & \hbox{if } B_{n+1} = 0 \\ 
				 L_{n}^{+}(\mathbf{y_1y_2}), & \hbox{if } B_{n+1} = 1 
			    \end{cases}
\end{equation*}
where
\begin{align*}
L_{n}^{-}(\mathbf{y_1y_2}) &= \displaystyle\frac{L_{n}(\mathbf{y_1}) + L_{n}(\mathbf{y_2})}{1 + L_{n}(\mathbf{y_1})L_{n}(\mathbf{y_2})}, \\
L_{n}^{+}(\mathbf{y_1y_2}) &= L_{n}(\mathbf{y_1})L_{n}(\mathbf{y_2}). 
\end{align*}
We denote $\Prob\left[ . \right] \triangleq \Prob\left[ . | W_{n}\right]$ for shorthand notation.
Let us define two auxiliary processes
\begin{align*}
 \Prob\left[ L_{n}(\mathbf{y}) \gneq 1 \right] &\triangleq \Prob\left[ L_{n}(\mathbf{y}) > 1 \right] + \frac{1}{2} \Prob\left[ L_{n}(\mathbf{y}) = 1 \right]\\
 \Prob\left[ L_{n}(\mathbf{y}) \lneq 1 \right] &\triangleq \Prob\left[ L_{n}(\mathbf{y}) < 1 \right] + \frac{1}{2} \Prob\left[ L_{n}(\mathbf{y}) = 1 \right]
\end{align*}
such that $\Prob\left[ L_{n}(\mathbf{y}) \gneq 1 \right]  + \Prob\left[ L_{n}(\mathbf{y}) \lneq 1 \right] = 1$.

The following two propositions investigate monotonicity properties of the processes $\Prob\left[ L_{n}(\mathbf{y}) \gneq 1 \right]$ and $\Prob\left[ L_{n}(\mathbf{y}) \lneq 1 \right]$. 
Their proofs will be carried together.
\begin{proposition}\label{prop::LR_monotonicity}
Given that $\Prob\left[L_n(\mathbf{y}) > 1\right] \leq \Prob\left[L_n(\mathbf{y}) < 1\right]$ holds for a particular $n\geq 1$, the polar transformations for the likelihood ratios satisfy
\begin{align*}
  \Prob\left[L_{n}^{+}(\mathbf{y_1y_2}) \gneq 1\right] \leq \Prob\left[L_n(\mathbf{y}) \gneq 1\right] \leq \Prob\left[L_{n}^{-}(\mathbf{y_1y_2}) \gneq 1\right]  \\
  \Prob\left[L_{n}^{-}(\mathbf{y_1y_2}) \lneq 1\right] \leq \Prob\left[L_n(\mathbf{y}) \lneq 1\right] \leq \Prob\left[L_{n}^{+}(\mathbf{y_1y_2}) \lneq 1\right]  
\end{align*}
\end{proposition}
\begin{proposition}\label{prop::LR_mass_preservation}
Given that $\Prob\left[L_n(\mathbf{y}) > 1\right] \leq \Prob\left[L_n(\mathbf{y}) < 1\right]$ holds for a particular $n\geq 1$, the basic polarization transformations preserve this inequality, i.e. at the next level we get
\begin{align*}
\Prob\left[L_{n}^{-}(\mathbf{y_1y_2}) > 1\right] &\leq \Prob\left[L_{n}^{-}(\mathbf{y_1y_2}) < 1\right], \\
\Prob\left[L_{n}^{+}(\mathbf{y_1y_2}) > 1\right] &\leq \Prob\left[L_{n}^{+}(\mathbf{y_1y_2}) < 1\right]. 
\end{align*}
\end{proposition}
\begin{proof}[Proof of Propositions \ref{prop::LR_monotonicity} and \ref{prop::LR_mass_preservation}]
We first derive some useful expressions for the quantities of interest. After applying the minus transformation, we get
\begin{multline}\label{eq::L_minus_less_than_one}
\Prob\left[L_{n}^{-}(\mathbf{y_1y_2}) < 1\right] = \Prob\left[L_n(\mathbf{y_1}) < 1 \right]\Prob\left[L_n(\mathbf{y_2})  < 1 \right] \\ 
+ \Prob\left[L_n(\mathbf{y_1})  > 1 \right]\Prob\left[L_n(\mathbf{y_2})  > 1\right]  \\
= \Prob\left[L_n(\mathbf{y})  < 1 \right]^{2} + \Prob\left[L_n(\mathbf{y}) > 1\right]^{2},
\end{multline}
\begin{multline}\label{eq::L_minus_greater_than_one}
\Prob\left[L_{n}^{-}(\mathbf{y_1y_2}) > 1\right] = \Prob\left[L_n(\mathbf{y_1}) < 1 \right]\Prob\left[L_n(\mathbf{y_2})  > 1 \right] \\
+ \Prob\left[L_n(\mathbf{y_1}) > 1 \right]\Prob\left[L_n(\mathbf{y_2})  < 1 \right] \\
= 2 \Prob\left[L_n(\mathbf{y})  < 1 \right] \Prob\left[L_n(\mathbf{y}) > 1\right],
\end{multline}
and 
\begin{multline}\label{eq::L_minus_equal_one}
\Prob\left[L_{n}^{-}(\mathbf{y_1y_2}) = 1\right] = \Prob\left[L_n(\mathbf{y_1}) = 1 \right] + \Prob\left[L_n(\mathbf{y_2})  = 1 \right] \\
- \Prob\left[L_n(\mathbf{y_1}) = 1 \right]\Prob\left[L_n(\mathbf{y_2}) = 1 \right] \\
= 2\Prob\left[L_n(\mathbf{y}) = 1 \right] - \Prob\left[L_n(\mathbf{y}) = 1\right]^{2}.
\end{multline}
Using \eqref{eq::L_minus_greater_than_one} and \eqref{eq::L_minus_equal_one}, we obtain similarly
\begin{multline}\label{eq::L_minus_greater_halfequal_than_one}
\Prob\left[L_{n}^{-}(\mathbf{y_1y_2}) \gneq 1\right] \\
= \Prob\left[L_{n}^{-}(\mathbf{y_1y_2}) > 1\right] + \frac{1}{2}\Prob\left[L_{n}^{-}(\mathbf{y_1y_2}) = 1\right] \\
= 2 \Prob\left[L_n(\mathbf{y}) \lneq 1 \right]\Prob\left[L_n(\mathbf{y}) \gneq 1 \right] 
\end{multline}
as by few simple manipulations we get
\begin{multline}\label{eq::L_lneq_time_gneq}
\Prob\left[L_n(\mathbf{y}) \lneq 1 \right]\Prob\left[L_n(\mathbf{y}) \gneq 1 \right] \\
= \Prob\left[ L_{n}(\mathbf{y}) < 1 \right]\Prob\left[ L_{n}(\mathbf{y}) > 1 \right] + \frac{1}{2}\Prob\left[ L_{n}(\mathbf{y}) = 1 \right] \times \\
\displaystyle\underbrace{\left(\Prob\left[ L_{n}(\mathbf{y}) < 1 \right] + \Prob\left[ L_{n}(\mathbf{y}) > 1 \right]\right)}_{1 - \Prob\left[ L_{n}(\mathbf{y}) = 1 \right]} 
+ \frac{1}{4}\Prob\left[ L_{n}(\mathbf{y}) = 1 \right]^2 \\
= \Prob\left[L_n(\mathbf{y})  < 1 \right] \Prob\left[L_n(\mathbf{y}) > 1\right] + \frac{1}{2} \Prob\left[L_n(\mathbf{y}) = 1 \right] \\
- \frac{1}{4}\Prob\left[L_n(\mathbf{y}) = 1\right]^{2}.
\end{multline}
Hence, we also have
\begin{equation}\label{eq::L_minus_less_halfequal_than_one}
\Prob\left[L_{n}^{-}(\mathbf{y_1y_2}) \lneq 1\right] = \Prob\left[L_n(\mathbf{y}) \lneq 1 \right]^{2} + \Prob\left[L_n(\mathbf{y}) \gneq 1\right]^{2}. 
\end{equation}
Noting the difference of the quantities in \eqref{eq::L_minus_less_halfequal_than_one} and \eqref{eq::L_minus_greater_halfequal_than_one} equals
\begin{equation}\label{eq::Qn_minus_squared}
\left(\Prob\left[L_n(\mathbf{y}) \lneq 1\right] - \Prob\left[L_n(\mathbf{y}) \gneq 1\right]\right)^{2} \geq 0,
\end{equation}
proves the claim of Proposition \ref{prop::LR_mass_preservation} for the minus transformation. 

On the other hand, by assumption $\Prob\left[L_n(\mathbf{y}) \lneq 1 \right]\in[0.5, 1]$ holds. So, we have 
\begin{multline*}
\Prob\left[L_{n}^{-}(\mathbf{y_1y_2}) \gneq 1\right] = 2 \Prob\left[L_n(\mathbf{y}) \lneq 1 \right]\Prob\left[L_n(\mathbf{y}) \gneq 1 \right] \\
\geq  \Prob\left[L_{n}(\mathbf{y}) \gneq 1\right],
\end{multline*}
which also implies 
\begin{equation*}
\Prob\left[L_{n}^{-}(\mathbf{y_1y_2}) \lneq 1\right] \leq  \Prob\left[L_{n}(\mathbf{y}) \lneq 1\right], 
\end{equation*}
proving the inequalities in Proposition \ref{prop::LR_monotonicity} for the minus transformation.

For the plus transformation, we use a property following the symmetry of the channels
\begin{equation}\label{eq::sym_cond}
W(y|0) = \displaystyle\frac{W(y|1)}{L(y)} \Rightarrow \Prob\left[L_{n}(\mathbf{y}) = \ell\right] = \displaystyle\frac{1}{\ell} \Prob\left[L_{n}(\mathbf{y})= \displaystyle\frac{1}{\ell}\right].
\end{equation}
Then, we can write
\begin{align} 
&\Prob\left[L_{n}^{+}(\mathbf{y_1y_2}) \lneq 1\right] \nonumber\\
= &\displaystyle\sum_{\ell_1 < 1}\sum_{\ell_2 < 1} \Prob\left[L_n(\mathbf{y_1}) = \ell_1\right]\Prob\left[L_n(\mathbf{y_2}) = \ell_2\right]  \nonumber\\
&+ \displaystyle\sum_{\ell_1 < 1}\sum_{1 \leq \ell_2 < 1/\ell_1} \Prob\left[L_n(\mathbf{y_1}) = \ell_1\right]\Prob\left[L_n(\mathbf{y_2}) = \ell_2\right]  \nonumber\\
&+ \displaystyle\sum_{\ell_1 \geq 1}\sum_{\ell_2 \leq 1/\ell_1} \Prob\left[L_n(\mathbf{y_1}) = \ell_1\right]\Prob\left[L_n(\mathbf{y_2}) = \ell_2\right]  \nonumber\\
&- \frac{1}{2}\Prob\left[L_{n}(\mathbf{y}) = 1\right]^2 \nonumber\\
= &\Prob\left[L_{n}(\mathbf{y}) < 1\right]^2 - \frac{1}{2}\Prob\left[L_{n}(\mathbf{y}) = 1\right]^2  \nonumber\\
&+ \displaystyle\sum_{\ell_1 > 1}\sum_{1 \leq \ell_2 < \ell_1} \ell_1\Prob\left[L_n(\mathbf{y_1})= \ell_1\right]\Prob\left[L_n(\mathbf{y_2}) = \ell_2\right] \nonumber \\
&+ \displaystyle\sum_{\ell_1 \geq 1}\sum_{\ell_2 \geq \ell_1} \ell_2\Prob\left[L_n(\mathbf{y_1}) = \ell_1\right]\Prob\left[L_n(\mathbf{y_2}) = \ell_2\right] \nonumber \\
=  &\Prob\left[L_{n}(\mathbf{y}) < 1\right]^2 - \frac{1}{2}\Prob\left[L_{n}(\mathbf{y}) = 1\right]^2 \nonumber \\ 
&+ \displaystyle\sum_{\ell_1 > 1}\sum_{1 < \ell_2 < \ell_1} \ell_1\Prob\left[L_n(\mathbf{y_1})= \ell_1\right]\Prob\left[L_n(\mathbf{y_2}) = \ell_2\right] \nonumber\\
&+ \Prob\left[L_{n}(\mathbf{y}) = 1\right]\displaystyle\sum_{\ell_1 > 1}\ell_1\Prob\left[L_n(\mathbf{y_1})= \ell_1\right] \nonumber \\
&+ \displaystyle\sum_{\ell_1 > 1}\sum_{\ell_2 \geq \ell_1} \ell_2\Prob\left[L_n(\mathbf{y_1}) = \ell_1\right]\Prob\left[L_n(\mathbf{y_2}) = \ell_2\right] \nonumber\\
&+ \Prob\left[L_{n}(\mathbf{y}) = 1\right]\displaystyle\sum_{\ell_2 > 1}\ell_2\Prob\left[L_n(\mathbf{y_1})= \ell_2\right] + \Prob\left[L_{n}(\mathbf{y}) = 1\right]^2 \nonumber \\
=  &\Prob\left[L_{n}(\mathbf{y}) \lneq 1\right]^2 \nonumber\\
\label{eq::L_plus_less_halfequal_than_one}&+ \displaystyle\sum_{\ell_1 \gneq 1}\sum_{\ell_2 \gneq 1} \Prob\left[L_n(\mathbf{y_1}) = \ell_1\right]\Prob\left[L_n(\mathbf{y_2}) = \ell_2\right] \max\{\ell_1, \ell_2\} 
\end{align}
where we abuse the notation to define
\begin{multline*}
\displaystyle\sum_{\ell_1 \gneq 1}\sum_{\ell_2 \gneq 1} \Prob\left[L_n(\mathbf{y_1}) = \ell_1\right]\Prob\left[L_n(\mathbf{y_2}) = \ell_2\right] \max\{\ell_1, \ell_2\}  \\
= \displaystyle\sum_{\ell_1 > 1}\sum_{\ell_2 > 1} \Prob\left[L_n(\mathbf{y_1}) = \ell_1\right]\Prob\left[L_n(\mathbf{y_2}) = \ell_2\right] \max\{\ell_1, \ell_2\}  \\
+ \Prob\left[L_{n}(\mathbf{y}) = 1\right]\displaystyle\sum_{\ell > 1}\ell\Prob\left[L_n(\mathbf{y_1})= \ell\right] + \frac{1}{4}\Prob\left[L_{n}(\mathbf{y}) = 1\right]^2. 
\end{multline*}
In the same spirit, we define
\begin{multline*}
\displaystyle\sum_{\ell_1 \gneq 1}\sum_{\ell_2 \gneq 1} \Prob\left[L_n(\mathbf{y_1}) = \ell_1\right]\Prob\left[L_n(\mathbf{y_2}) = \ell_2\right] \min\{\ell_1, \ell_2\}  \\
= \displaystyle\sum_{\ell_1 > 1}\sum_{\ell_2 > 1} \Prob\left[L_n(\mathbf{y_1}) = \ell_1\right]\Prob\left[L_n(\mathbf{y_2}) = \ell_2\right] \min\{\ell_1, \ell_2\}  \\
+ \Prob\left[L_{n}(\mathbf{y}) = 1\right]\Prob\left[L_{n}(\mathbf{y}) > 1\right] + \frac{1}{4}\Prob\left[L_{n}(\mathbf{y}) = 1\right]^2,
\end{multline*}
and we note that
\begin{multline}\label{eq::min_plus_max}
\displaystyle\sum_{\ell_1 \gneq 1}\sum_{\ell_2 \gneq 1} \Prob\left[L_n(\mathbf{y_1}) = \ell_1\right]\Prob\left[L_n(\mathbf{y_2}) = \ell_2\right] \\
\left(\max\{\ell_1, \ell_2\} + \min\{\ell_1, \ell_2\} \right) \\
= \displaystyle\sum_{\ell_1 \gneq 1}\sum_{\ell_2 \gneq 1} \Prob\left[L_n(\mathbf{y_1}) = \ell_1\right]\Prob\left[L_n(\mathbf{y_2}) = \ell_2\right] (\ell_1 + \ell_2) \\
= 2\displaystyle\sum_{\ell_1 \gneq 1} \Prob\left[L_n(\mathbf{y_1}) = \ell_1\right] \ell_1 \sum_{\ell_2 \gneq 1} \Prob\left[L_n(\mathbf{y_2}) = \ell_2\right] \\
= 2\Prob\left[L_{n}(\mathbf{y}) \lneq 1\right]\Prob\left[L_{n}(\mathbf{y}) \gneq 1\right].
\end{multline}
As 
\begin{multline*}
1 = \Prob\left[L_{n}^{+}(\mathbf{y_1y_2}) \lneq 1\right] + \Prob\left[L_{n}^{+}(\mathbf{y_1y_2}) \gneq 1\right] \\
=  \left(\Prob\left[L_{n}(\mathbf{y}) \lneq 1\right] + \Prob\left[L_{n}(\mathbf{y}) \gneq 1\right]\right)^2 \\
\end{multline*}
must hold, we get
\begin{multline} \label{eq::L_plus_greater_halfequal_than_one}
\Prob\left[L_{n}^{+}(\mathbf{y_1y_2}) \gneq 1\right] = \Prob\left[L_{n}(\mathbf{y}) \gneq 1\right]^2 \\
+ \displaystyle\sum_{\ell_1 \gneq 1}\sum_{\ell_2 \gneq 1} \Prob\left[L_n(\mathbf{y_1}) = \ell_1\right]\Prob\left[L_n(\mathbf{y_2}) = \ell_2\right] \min\{\ell_1, \ell_2\} 
\end{multline}

Therefore, \eqref{eq::L_plus_less_halfequal_than_one} and \eqref{eq::L_plus_greater_halfequal_than_one} proves that 
\begin{equation*}
\Prob\left[L_{n}^{+}(\mathbf{y_1y_2}) \lneq 1\right] \geq \Prob\left[L_{n}^{+}(\mathbf{y_1y_2}) \gneq 1\right]
\end{equation*}
holds as claimed by Proposition \ref{prop::LR_mass_preservation}.  

On the other hand, we can decompose $\Prob\left[L_{n}(\mathbf{y}) \gneq 1\right]$ into
\begin{multline}\label{eq::L_gneq_one_decomposition}
\Prob\left[L_{n}(\mathbf{y}) \gneq 1\right] = \Prob\left[L_{n}(\mathbf{y}) > 1\right]^2  \\
+ \Prob\left[L_{n}(\mathbf{y}) > 1\right] \Prob\left[L_{n}(\mathbf{y}) \leq 1\right]  + \frac{1}{2}\Prob\left[L_{n}(\mathbf{y}) =1 \right] \\
=\left(\Prob\left[L_{n}(\mathbf{y}) \gneq 1\right] - \frac{1}{2}\Prob\left[L_{n}(\mathbf{y}) = 1\right]\right)^2  \\
+ \Prob\left[L_{n}(\mathbf{y}) > 1\right] \Prob\left[L_{n}(\mathbf{y}) \leq 1\right]  + \frac{1}{2}\Prob\left[L_{n}(\mathbf{y}) =1 \right] \\
= \Prob\left[L_{n}(\mathbf{y}) \gneq 1\right]^2 + \Prob\left[L_{n}(\mathbf{y}) > 1\right] \Prob\left[L_{n}(\mathbf{y}) < 1\right] \\
+ \frac{1}{2}\Prob\left[L_{n}(\mathbf{y}) =1 \right] - \frac{1}{4}\Prob\left[L_{n}(\mathbf{y}) = 1\right]^2 \\
= \Prob\left[L_{n}(\mathbf{y}) \gneq 1\right]^2 + \Prob\left[L_{n}(\mathbf{y}) \lneq 1\right]\Prob\left[L_{n}(\mathbf{y}) \gneq 1\right]
\end{multline}
where we used the derivation in \eqref{eq::L_lneq_time_gneq} to get the final equality. Comparing the expressions in \eqref{eq::L_plus_greater_halfequal_than_one} and \eqref{eq::L_gneq_one_decomposition} 
in the light of \eqref{eq::min_plus_max}, we see that
\begin{equation*}
 \Prob\left[L_{n}^{+}(\mathbf{y_1y_2}) \gneq 1\right] \leq \Prob\left[L_{n}(\mathbf{y}) \gneq 1\right],
\end{equation*}
which also implies
\begin{equation*}
 \Prob\left[L_{n}^{+}(\mathbf{y_1y_2}) \lneq 1\right] \geq \Prob\left[L_{n}(\mathbf{y}) \lneq 1\right].
\end{equation*}
proving the claimed inequalities in Proposition \ref{prop::LR_monotonicity} for the plus transformation.
\end{proof}

Next, we show the average of the transformed plus and minus quantities also satisfy some monotonicity properties.
\begin{proposition}\label{prop::one_step_ineq}
The following set of inequalities hold:
\begin{align}
 \label{eq::L_gneq_ineq}\Prob\left[L_{n}^{-}(\mathbf{y_1y_2}) \gneq 1\right] + \Prob\left[L_{n}^{+}(\mathbf{y_1y_2}) \gneq 1\right] \geq 2\Prob\left[L_{n}(\mathbf{y}) \gneq 1\right], \\
 \label{eq::L_lneq_ineq}\Prob\left[L_{n}^{-}(\mathbf{y_1y_2}) \lneq 1\right] + \Prob\left[L_{n}^{+}(\mathbf{y_1y_2}) \lneq 1\right] \leq 2\Prob\left[L_{n}(\mathbf{y}) \lneq 1\right], \\
 \label{eq::L_equal_one_ineq}\Prob\left[L_{n}^{-}(\mathbf{y_1y_2}) = 1\right] + \Prob\left[L_{n}^{+}(\mathbf{y_1y_2}) = 1\right] \geq 2\Prob\left[L_{n}(\mathbf{y}) = 1\right].
\end{align}
Hence, we also have
\begin{align*}
 \Prob\left[L_{n}^{-}(\mathbf{y_1y_2}) < 1\right] + \Prob\left[L_{n}^{+}(\mathbf{y_1y_2}) < 1\right] \leq 2\Prob\left[L_{n}(\mathbf{y}) < 1\right], \\
 \Prob\left[L_{n}^{-}(\mathbf{y_1y_2}) \geq 1\right] + \Prob\left[L_{n}^{+}(\mathbf{y_1y_2}) \geq 1\right] \geq 2\Prob\left[L_{n}(\mathbf{y}) \geq 1\right].
\end{align*}
\end{proposition}
\begin{proof}
We start by proving the inequality in \eqref{eq::L_gneq_ineq}.
Using the expressions derived in \eqref{eq::L_minus_greater_halfequal_than_one} and \eqref{eq::L_plus_greater_halfequal_than_one} show that
\begin{multline*}
 \Prob\left[L_{n}^{-}(\mathbf{y_1y_2}) \gneq 1\right] + \Prob\left[L_{n}^{+}(\mathbf{y_1y_2}) \gneq 1\right] \\
= 2\Prob\left[L_{n}(\mathbf{y}) \gneq 1\right]\underbrace{\Prob\left[L_{n}(\mathbf{y}) \lneq 1\right]}_{1 - \Prob\left[L_{n}(\mathbf{y}) \gneq 1\right]} + \Prob\left[L_{n}(\mathbf{y}) \gneq 1\right]^2 \\
+ \displaystyle\sum_{\ell_1 \gneq 1}\sum_{\ell_2 \gneq 1} \Prob\left[L_n(\mathbf{y_1})  = \ell_1\right]\Prob\left[L_n(\mathbf{y_2}) = \ell_2\right] \min\{\ell_1, \ell_2\} \\
= 2\Prob\left[L_{n}(\mathbf{y}) \gneq 1\right] - \Prob\left[L_{n}(\mathbf{y}) \gneq 1\right]^2 \\
+ \displaystyle\sum_{\ell_1 \gneq 1}\sum_{\ell_2 \gneq 1} \Prob\left[L_n(\mathbf{y_1})  = \ell_1\right]\Prob\left[L_n(\mathbf{y_2}) = \ell_2\right] \min\{\ell_1, \ell_2\} \\
\geq 2\Prob\left[L_{n}(\mathbf{y}) \gneq 1\right] 
\end{multline*}
where the inequality follows from
\begin{multline}\label{eq::non_neg}
 \displaystyle\sum_{\ell_1 \gneq 1}\sum_{\ell_2 \gneq 1} \Prob\left[L_n(\mathbf{y_1})  = \ell_1\right]\Prob\left[L_n(\mathbf{y_2}) = \ell_2\right] \min\{\ell_1, \ell_2\} \\
\geq \Prob\left[L_{n}(\mathbf{y}) \gneq 1\right]^2.
\end{multline}
This also proves the inequality in \eqref{eq::L_lneq_ineq} in view of the relation $\Prob\left[L_{n+1}(\mathbf{y_1y_2}) \lneq 1\right] = 1 - \Prob\left[L_{n+1}(\mathbf{y_1y_2}) \gneq 1\right]$.
Finally, to prove \eqref{eq::L_equal_one_ineq}, we write
\begin{multline*}
\Prob\left[L_{n}^{-}(\mathbf{y_1y_2}) = 1\right] + \Prob\left[L_{n}^{+}(\mathbf{y_1y_2}) = 1\right] \\
\geq 2 \Prob\left[L_{n}(\mathbf{y}) = 1\right]- \Prob\left[L_{n}(\mathbf{y}) = 1\right]^2 + \Prob\left[L_{n}(\mathbf{y}) = 1\right]^2
\end{multline*}
where we used \eqref{eq::L_minus_equal_one} and simply noted that $\Prob\left[L_{n}^{+}(\mathbf{y_1y_2}) = 1\right] \geq \Prob\left[L_{n}(\mathbf{y}) = 1\right]^2$ holds.
\end{proof}

Before we discuss the implications of the inequalities in Proposition \ref{prop::one_step_ineq} on the processes, we define another channel parameter as
\begin{align*}
 Q_{n} &\triangleq \Prob\left[ L_{n}(\mathbf{y}) \lneq 1 \right] - \Prob\left[ L_{n}(\mathbf{y}) \gneq 1 \right] \\
 &= \Prob\left[ L_{n}(\mathbf{y}) < 1 \right] - \Prob\left[ L_{n}(\mathbf{y}) > 1 \right].
\end{align*}

The one step transformations of $Q_{n}$ are given by
\begin{proposition}\label{prop::Q_n_one_step_evolution}
\begin{equation*}
 Q_{n+1} = \begin{cases}
				 Q_{n}^{-}, & \hbox{if } B_{n+1} = 0 \\ 
				 Q_{n}^{+}, & \hbox{if } B_{n+1} = 1 
			    \end{cases}
\end{equation*} 
where 
\begin{align*}
 Q_{n}^{-} &= Q_{n}^2, \\
 Q_{n}^{+} &\in \left[Q_{n}, 2Q_{n} - Q_{n}^2\right]. 
\end{align*}
\end{proposition}
\begin{proof}
From the derivation of \eqref{eq::Qn_minus_squared}, we immediately get $Q_{n}^{-} = Q_{n}^2$.
Moreover, Proposition \ref{prop::LR_monotonicity} implies $Q_{n}^{+} \geq Q_{n}$. On the other hand, using \eqref{eq::L_plus_less_halfequal_than_one} and \eqref{eq::L_plus_greater_halfequal_than_one} we have
\begin{multline*}
 Q_{n}^{+} = \Prob\left[ L_{n}(\mathbf{y}) \lneq 1 \right]^2 - \Prob\left[ L_{n}(\mathbf{y}) \gneq 1 \right]^2 \\
 + \displaystyle\sum_{\ell_1 \gneq 1}\sum_{\ell_2 \gneq 1} \Prob\left[L_n(\mathbf{y_1}) = \ell_1\right]\Prob\left[L_n(\mathbf{y_2}) = \ell_2\right] \times \\
\left(\max\{\ell_1, \ell_2\} - \min\{\ell_1, \ell_2\}\right) \\
= Q_{n} + \displaystyle\sum_{\ell_1 \gneq 1}\sum_{\ell_2 \gneq 1} \Prob\left[L_n(\mathbf{y_1}) = \ell_1\right]\Prob\left[L_n(\mathbf{y_2}) = \ell_2\right] \times \\
\left(\max\{\ell_1, \ell_2\} - \min\{\ell_1, \ell_2\}\right).
\end{multline*}
Moreover, note that
\begin{multline*}
 Q_{n} - Q_{n}^2 = Q_{n} (1- Q_{n}) \\
= \left( \Prob\left[ L_{n}(\mathbf{y}) \lneq 1 \right] - \Prob\left[ L_{n}(\mathbf{y}) \gneq 1 \right]\right) 2 \Prob\left[ L_{n}(\mathbf{y}) \gneq 1 \right] \\
= 2 \Prob\left[ L_{n}(\mathbf{y}) \gneq 1 \right]\Prob\left[ L_{n}(\mathbf{y}) \lneq 1 \right] - 2\Prob\left[ L_{n}(\mathbf{y}) \gneq 1 \right]^2 
\end{multline*}
as $1 - Q_{n} = 2\Prob\left[ L_{n}(\mathbf{y}) \gneq 1 \right]$.
Therefore,
\begin{multline*}
 2Q_{n} - Q_{n}^2 - Q_{n}^{+} \\
= 2 \Prob\left[ L_{n}(\mathbf{y}) \gneq 1 \right]\Prob\left[ L_{n}(\mathbf{y}) \lneq 1 \right] - 2\Prob\left[ L_{n}(\mathbf{y}) \gneq 1 \right]^2 + Q_{n} \\
- Q_{n} - \displaystyle\sum_{\ell_1 \gneq 1}\sum_{\ell_2 \gneq 1} \Prob\left[L_n(\mathbf{y_1}) = \ell_1\right]\Prob\left[L_n(\mathbf{y_2}) = \ell_2\right]\times \\
\left(\max\{\ell_1, \ell_2\} - \min\{\ell_1, \ell_2\}\right).
\end{multline*}
Now, using the expression in \eqref{eq::min_plus_max} instead of $2 \Prob\left[ L_{n}(\mathbf{y}) \gneq 1 \right]\Prob\left[ L_{n}(\mathbf{y}) \lneq 1 \right]$ we  get
\begin{multline*}
 2Q_{n} - Q_{n}^2 - Q_{n}^{+} =  - 2\Prob\left[ L_{n}(\mathbf{y}) \gneq 1 \right]^2 \\
 + \displaystyle\sum_{\ell_1 \gneq 1}\sum_{\ell_2 \gneq 1} \Prob\left[L_n(\mathbf{y_1}) = \ell_1\right]\Prob\left[L_n(\mathbf{y_2}) = \ell_2\right] \times \\
\left(\max\{\ell_1, \ell_2\} + \min\{\ell_1, \ell_2\}\right) \\
 - \displaystyle\sum_{\ell_1 \gneq 1}\sum_{\ell_2 \gneq 1} \Prob\left[L_n(\mathbf{y_1}) = \ell_1\right]\Prob\left[L_n(\mathbf{y_2}) = \ell_2\right] \times \\
\left(\max\{\ell_1, \ell_2\} - \min\{\ell_1, \ell_2\}\right) \\
 = 2\displaystyle\sum_{\ell_1 \gneq 1}\sum_{\ell_2 \gneq 1} \Prob\left[L_n(\mathbf{y_1}) = \ell_1\right]\Prob\left[L_n(\mathbf{y_2}) = \ell_2\right] \min\{\ell_1, \ell_2\} \\
- 2\Prob\left[ L_{n}(\mathbf{y}) \gneq 1 \right]^2 \geq 0
\end{multline*}
where the non-negativity is due to \eqref{eq::non_neg} once again.
\end{proof}
\begin{corollary}
 The BEC is an extremal channel in the evolution of the process $Q_{n}$.
\end{corollary}
\begin{proof}
 The proof follows by noting that being a BEC is preserved under the polarization transformations~\cite{1669570} with $Q_{n}^{+} = 2Q_{n} - Q_{n}^2$. 
\end{proof}

Now, we discuss the convergence properties of the processes we considered so far.

\begin{proposition}\label{prop::martingales_convergences}
Let $W$ be a symmetric B-DMC such that $\Prob[L_{0}(\mathbf{y}) > 1] \leq \Prob[L_{0}(\mathbf{y}) < 1]$ holds. Then,
\begin{enumerate}
\item[(i)] The process $Q_n$ is a bounded supermartingale in $[0, 1]$ and converges a.s. to $\{0, 1\}$.
\item[(ii)] The process $\Prob\left(L_{n}(\mathbf{y}) \gneq 1\right)$ is a bounded submartingale in $[0, 0.5]$ and converges a.s. to $\{0, 0.5\}$. 
\item[(iii)] The process $\Prob\left(L_{n}(\mathbf{y})= 1\right)$ is a bounded submartingale in $\in[0, 1]$  and converges a.s. to $\{0, 1\}$.  
\item[(iv)] The process $\Prob\left(L_{n}(\mathbf{y}) \lneq 1\right)$ is a bounded supermartingale in $\in[0.5, 1]$ and converges a.s. to $\{0.5, 1\}$.
\end{enumerate}
\end{proposition}
\begin{proof}
The assumption on the channel $W$ implies via Proposition \ref{prop::LR_monotonicity} that $\Prob\left[L_{n}(\mathbf{y}) > 1\right] \leq \Prob\left[L_{n}(\mathbf{y}) < 1\right]$ holds for all $n = 1, 2, \ldots$.
This constraints the probabilities to $\Prob\left[L_{n}(\mathbf{y}) > 1\right] \in[0, 0.5]$, $\Prob\left[L_{n}(\mathbf{y}) < 1\right] \in[0.5, 1]$, $\Prob\left[L_{n}(\mathbf{y}) = 1\right] \in[0, 1]$, 
from which the boundedness statements follow. 

The inequalities proved in Proposition \ref{prop::one_step_ineq} shows the processes are the claimed martingales. From general results on bounded martingales, it follows the processes converge a.s. 
The only part left is to prove the convergence is to the extremes of the bounded intervals.
For the process $Q_{n}$, we know by Proposition \ref{prop::Q_n_one_step_evolution} that $Q_{n}^{-} = Q_{n}^2$. 
One can complete the proof that $Q_{n}$ converges to the extremes using this relation in a similar fashion as in the proof of \cite[Proposition 9]{1669570}
of the convergence to the extremes of the Bhattacharyya process of the synthetic channels associated with the polarization transformations:
 \begin{align}
 &\Expt{\left[ \lvert Q_{n+1} - Q_{n} \rvert \right] } \xrightarrow[n\to\infty]{} 0 \nonumber \\
 \Rightarrow &\Expt{\left[ \lvert Q_{n+1} - Q_{n}\rvert \right] } \geq \frac{1}{2} \Expt{\left[ Q_{n} \left( 1-  Q_{n}\right) \right] }  \xrightarrow[n\to\infty]{} 0, \nonumber  
 \end{align} 
 whence $Q_{\infty}\in\{0, 1\}$. Similarly, we know by \eqref{eq::L_minus_equal_one} 
 that $\Prob\left[L_{n}^{-}(\mathbf{y_1y_2}) = 1\right] = 2\Prob\left[L_{n}(\mathbf{y}) = 1\right] - \Prob\left[L_{n}(\mathbf{y}) = 1\right]^2$ holds, 
 so that once again $\Prob\left[L_{\infty}(\mathbf{y}) = 1\right]\in\{0, 1\}$ since
 \begin{align*}
 &\Expt{\left[ \lvert \Prob\left[L_{n+1}(\mathbf{y_1y_2}) = 1\right] - \Prob\left[L_{n}(\mathbf{y}) = 1\right] \rvert \right] } \xrightarrow[n\to\infty]{} 0   \\
 \Rightarrow &\Expt{\left[ \lvert \Prob\left[L_{n+1}(\mathbf{y_1y_2}) = 1\right]  - \Prob\left[L_{n}(\mathbf{y}) = 1\right]\rvert \right] } \\
 &\geq \frac{1}{2}  \Expt{\left[ \Prob\left[L_{n}(\mathbf{y}) = 1\right] \left( 1-  \Prob\left[L_{n}(\mathbf{y}) = 1\right] \right) \right] }  \xrightarrow[n\to\infty]{} 0.  
 \end{align*} 
Now once $Q_{n}$ and $\Prob\left[L_{\infty}(\mathbf{y}) = 1\right]$ converge to their extremes, the remaining probabilities can only converge to the extremes claimed by the proposition.
\end{proof}

\subsection*{Channel Polarization Revisited}
Now, we revisit the theory of channel polarization for symmetric B-DMCs. 
Let us start by describing a perfect channel and a completely noisy channel in terms of the channel parameters we have discussed so far. 
It is easy to see that the channel is perfect when $Q(W) = 1$, which is possible only when $\Prob[L_{\infty}<1]= 1$, $\Prob[L_{\infty}=1] = \Prob[L_{\infty}>1] = 0$ hold. Without any surprise, we get $I(W) = 1$ in this case.
On the other hand, the channel is completely noisy when $\Prob[L_{\infty}=1] = 1$, $\Prob[L_{\infty}<1] = \Prob[L_{\infty}>1] = 0$, giving $Q(W) = 0$ and $I(W) = 0$. 

At this point, we can simply eliminate the other possibilities as we know $I_{n}$ is a bounded martingale process with $I_{\infty}\in\{0, 1\}$ from \cite{1669570} and capacity cannot be created. 
These are exactly the arguments proving once channels are polarized the fraction of moderate channels vanishes. Yet, let us ignore this knowledge for a moment to simply look to the four 
possible combinations of the pair $Q_{\infty}$ and $\Prob\left[L_{\infty} = 1\right]$, two of which we hopefully `never' end up with. 

\begin{enumerate}
 \item $Q_{\infty} = 1$, $\Prob[L_{\infty}=1] = 1$: As $Q_{\infty} = \Prob[L_{\infty}<1] - \Prob[L_{\infty}>1] = 1$ holds, we find $\Prob[L_{\infty}<1] = 1$, contradicting $\Prob[L_{\infty}=1] = 1$. So, this case is not possible. 
 \item $Q_{\infty} = 1$, $\Prob[L_{\infty}=1] = 0$: We look at a perfect channel.
 \item $Q_{\infty} = 0$, $\Prob[L_{\infty}=1] = 1$: We look at a completely noisy channel.
 \item $Q_{\infty} = 0$, $P\Prob[L_{\infty}=1] = 0$: These constraints only tell us $\Prob[L_{\infty}<1] = \Prob[L_{\infty}>1] = 0.5$ and $\Prob[L_{\infty}=1] = 0$. Hence, we are looking at a 
'completely moderate' channel. However, Proposition \ref{prop::LR_monotonicity} shows that the polar transforms are monotone for the probabilities of the likelihood ratios.
Consequently, this case will not occur unless we start with a channel at the state $\Prob[L_{0} < 1] = \Prob[L_{0} > 1] = 0.5$, but this would violate the symmetry condition. 
\end{enumerate}
Note that we still need the preservation of the sum capacities, i.e. $I(W_{n})$ being a martingale, to show that the fraction of perfect channels is $I(W)$. 

Moreover, the results on the rate of convergence of polar codes \cite{5205856} can be stated in terms of $Q_{n}$: note that the conditions (z.1), (z.2), (z.3) in \cite{5205856} still hold with $Z_{n}$ replaced by $Q_{n}$, and with the condition $\Prob\left[Z_\infty = 0\right] = I_0$ in (z.3) replaced by $\Prob\left[Q_\infty = 0\right] = 1 - I_0$. 

\subsection*{Properties of an approximation to the polar transforms}
In this section, we discuss the performance of an approximation to the minus transformation which appears in \cite{5946819} and \cite{Mine:ISITA12}. 
The min-sum approximation, as called in the literature, is defined as
\begin{equation}\label{eq::minus_approx}
 \log L_{n}^{-}(\mathbf{y_1y_2}) = -\sign(\ell_{1} * \ell_{2}) \min\{|\ell_{1}|, |\ell_{2}|\}
\end{equation}
where $\ell_{1} \triangleq \log L_{n}(\mathbf{y_1}), \ell_{2} \triangleq \log L_{n}(\mathbf{y_2})$.
While proposed in \cite{5946819} for efficient hardware implementations of polar codes, \cite{Mine:ISITA12} considers the performance of mismatched polar codes designed using the approximation over binary symmetric channels (BSC). 

First, we argue some of the derivations of the previous section extend as well to the approximate process defined as
\begin{equation*}
 \tilde{L}_{n+1}(\mathbf{y_1y_2}) = \begin{cases}
				 \tilde{L}_{n}^{-}(\mathbf{y_1y_2}), & \hbox{if } B_{n+1} = 0 \\ 
				 \tilde{L}_{n}^{+}(\mathbf{y_1y_2}), & \hbox{if } B_{n+1} = 1 
			    \end{cases}
\end{equation*}
where
\begin{multline*}
\tilde{L}_{n}^{-}(\mathbf{y_1y_2}) = \exp\left\lbrace-\sign\left(\log\tilde{L}_{n}(\mathbf{y_1}) * \log\tilde{L}_{n}(\mathbf{y_2})\right)\right. \times\\
\left.\min\left\lbrace \lvert\log\tilde{L}_{n}(\mathbf{y_1})|, |\log\tilde{L}_{n}(\mathbf{y_2}) \rvert\right\rbrace\right\rbrace 
\end{multline*}
\begin{equation*}
\tilde{L}_{n}^{+}(\mathbf{y_1y_2}) = \tilde{L}_{n}(\mathbf{y_1})\tilde{L}_{n}(\mathbf{y_2}) 
\end{equation*}

This is explained by the fact that the approximate minus transformation of the likelihood ratios satisfy, as the exact case, the following properties:
\begin{align*}
\hbox{1) }
&\left\lbrace\tilde{L}_{n+1}(\mathbf{y_1y_2}) > 1\right\rbrace \\
&\iff \left\lbrace \tilde{L}_{n}(\mathbf{y_1}) > 1\right\rbrace \cap \left\lbrace \tilde{L}_{n}(\mathbf{y_2}) < 1\right\rbrace \\
&\hspace*{5mm}\bigcup \left\lbrace \tilde{L}_{n}(\mathbf{y_1}) < 1\right\rbrace \cap \left\lbrace \tilde{L}_{n}(\mathbf{y_2}) > 1\right\rbrace \\
\hbox{2) } 
&\left\lbrace\tilde{L}_{n+1}(\mathbf{y_1y_2}) = 1\right\rbrace, \\ 
&\iff \left\lbrace \tilde{L}_{n}(\mathbf{y_1}) = 1\right\rbrace\ \bigcup \left\lbrace \tilde{L}_{n}(\mathbf{y_2}) = 1\right\rbrace,  \\
\hbox{3) } 
&\left\lbrace\tilde{L}_{n+1}(\mathbf{y_1y_2}) < 1\right\rbrace \\
&\iff \left\lbrace \tilde{L}_{n}(\mathbf{y_1}) > 1\right\rbrace \cap \left\lbrace \tilde{L}_{n}(\mathbf{y_2}) > 1\right\rbrace \\
&\hspace*{5mm}\bigcup \left\lbrace \tilde{L}_{n}(\mathbf{y_1}) < 1\right\rbrace \cap \left\lbrace \tilde{L}_{n}(\mathbf{y_2}) < 1\right\rbrace. \\
\end{align*}
Hence, the below counterparts to \eqref{eq::L_minus_equal_one}, \eqref{eq::L_minus_greater_halfequal_than_one} and \eqref{eq::L_minus_less_halfequal_than_one} continue to hold.
\begin{align}
\label{eq::L_approx_minus_equal_one}\Prob\left[\tilde{L}_{n}^{-}(\mathbf{y_1y_2}) = 1\right] &= 2\Prob\left[\tilde{L}_n(\mathbf{y}) = 1 \right] - \Prob\left[\tilde{L}_n(\mathbf{y}) = 1 \right]^2,\\
\label{eq::L_approx_minus_greater_halfequal_than_one} \Prob\left[\tilde{L}_{n}^{-}(\mathbf{y_1y_2}) \gneq 1\right] &= 2\Prob\left[\tilde{L}_n(\mathbf{y}) \lneq 1 \right]\Prob\left[\tilde{L}_n(\mathbf{y}) \gneq 1\right], \\
\label{eq::L_approx_minus_less_halfequal_than_one} \Prob\left[\tilde{L}_{n}^{-}(\mathbf{y_1y_2}) \lneq 1\right] &= \Prob\left[\tilde{L}_n(\mathbf{y}) \lneq 1 \right]^{2} + \Prob\left[\tilde{L}_n(\mathbf{y}) \gneq 1\right]^{2}.
\end{align}
Similarly, for the plus transformation as the symmetry in the likelihood ratios is preserved by the approximation, one can use the LHS of \eqref{eq::sym_cond} to derive the below counterparts to
\eqref{eq::L_plus_less_halfequal_than_one} and \eqref{eq::L_plus_greater_halfequal_than_one}:
\begin{multline}\label{eq::L_approx_plus_less_halfequal_than_one}
\Prob\left[\tilde{L}_{n}^{+}(\mathbf{y_1y_2}) \lneq 1\right] = \Prob\left[L_{n}(\mathbf{y}) \lneq 1\right]^2 \\
+ \displaystyle\sum_{\begin{subarray}{c}
\mathbf{y_1y_2}: \\ \tilde{L}_n(\mathbf{y_1}) \gneq 1 \\ \tilde{L}_n(\mathbf{y_2}) \gneq 1 \end{subarray}} W(\mathbf{y_1}|\mathbf{0})W(\mathbf{y_2}|\mathbf{0}) \max\{L_n(\mathbf{y_1}), L_n(\mathbf{y_2})\} 
\end{multline}
\begin{multline}\label{eq::L_approx_plus_greater_halfequal_than_one}
 \Prob\left[\tilde{L}_{n}^{+}(\mathbf{y_1y_2}) \gneq 1\right] = \Prob\left[L_{n}(\mathbf{y}) \lneq 1\right]^2 \\
+ \displaystyle\sum_{\begin{subarray}{c}
\mathbf{y_1y_2}: \\ \tilde{L}_n(\mathbf{y_1}) \gneq 1 \\ \tilde{L}_n(\mathbf{y_2}) \gneq 1 \end{subarray}} W(\mathbf{y_1}|\mathbf{0})W(\mathbf{y_2}|\mathbf{0}) \min\{L_n(\mathbf{y_1}), L_n(\mathbf{y_2})\} 
\end{multline}

As a result, one can carry the proofs of Propositions \ref{prop::LR_monotonicity} and \ref{prop::LR_mass_preservation} in exactly the same way by replacing
the uses of  \eqref{eq::L_minus_equal_one}, \eqref{eq::L_minus_greater_halfequal_than_one}, \eqref{eq::L_minus_less_halfequal_than_one}, \eqref{eq::L_plus_less_halfequal_than_one}, and \eqref{eq::L_plus_greater_halfequal_than_one}
by \eqref{eq::L_approx_minus_equal_one}, \eqref{eq::L_approx_minus_greater_halfequal_than_one}, \eqref{eq::L_approx_minus_less_halfequal_than_one}, \eqref{eq::L_approx_plus_less_halfequal_than_one}, and \eqref{eq::L_approx_plus_greater_halfequal_than_one}, respectively. On the other hand, for a given 
$L(\mathbf{y_2}) \neq 1$, while the exact minus transformation is strictly monotone in $L(\mathbf{y_1})$ (increasing or decreasing), the approximate one is no longer strictly but simply monotone. So, one particular difference caused by the minus approximation is identical likelihood ratios obtained for some outputs which would otherwise be different from each others. Hence, following the approximation a plus transformation at the next level will result in more outputs having likelihood ratios equal to one.
Whether ultimately this would cause loss in the performance is an open problem, i.e. we do not know if
\begin{equation}\label{eq::joint_extreme}
\Prob\left[ L_{\infty}(\mathbf{y}) < 1 | \tilde{L}_{\infty}(\mathbf{y}) = 1\right] = 1
\end{equation}
is possible.

A sufficient condition to avoid the above situation from happening is the following:
\begin{multline}\label{eq::cond_no_loss}
 \hbox{If}\hspace{2mm} \left\lbrace\begin{array}{l}
            \{ L_{n}(\mathbf{y}) < 1 \} = \{ \tilde{L}_{n}(\mathbf{y}) < 1 \}, \\
 	   \{ L_{n}(\mathbf{y}) > 1 \} = \{ \tilde{L}_{n}(\mathbf{y}) > 1 \}
           \end{array}\right\rbrace \\
 \Rightarrow \left\lbrace\begin{array}{l}
            \{ L_{n+1}(\mathbf{y_1y_2}) < 1 \} = \{ \tilde{L}_{n+1}(\mathbf{y_1y_2}) < 1 \}, \\
 	   \{ L_{n+1}(\mathbf{y_1y_2}) > 1 \} = \{ \tilde{L}_{n+1}(\mathbf{y_1y_2}) > 1 \}
           \end{array}\right\rbrace.
\end{multline}
Consequently, no performance degradation would be incurred by such an approximation. 
The process $\tilde{Q}_{n}$ would behave exactly as the process $Q_{n}$, and the synthetic channels created by the approximate transformations would also polarize with $\tilde{Q}_{\infty}\in\{0, 1\}$.

Now, we discuss how the min-sum approximation can be modified to attain this goal.
The idea is to slightly perturb the identical likelihood ratios forced by the approximation to distinct values while keeping the symmetry, and the order 
\begin{equation*}
\hbox{If } L(\mathbf{y_1}) < L(\mathbf{y_2}) \Rightarrow \tilde{L}(\mathbf{y_1}) < \tilde{L}(\mathbf{y_2}), \quad \forall \mathbf{y_1}, \mathbf{y_2}.
\end{equation*}
In this case, this new version of the approximation would satisfy \eqref{eq::cond_no_loss}. The real trouble might be to find such an approximation `better' than the exact case. Still, as we simply want to avoid \eqref{eq::joint_extreme}, trading-off the order preservation requirement, the slight perturbations might still prevent the fraction of likelihood ratios of value $1$ to dominate the approximate case as opposed to the exact case in some of the synthetic channels.

\section{conclusions}
In this paper we investigated, starting with a symmetric B-DMC, the evolution of various probabilities related to the likelihood ratios of 
the synthetic channels created by the recursive application of the polarization transformations.
We showed the processes are bounded martingales converging to the extremes of the bounded intervals using similar proof techniques used in \cite{1669570} and the inherent symmetry in the channels. 
The analysis helped us to consider the approximation given in Equation \eqref{eq::minus_approx} for the likelihood ratio recursion.

\section{Acknowledgments}
This work was supported by Swiss National Science Foundation under grant number 200021-125347/1. 

\end{document}